\documentclass[12pt]{iopart}
\usepackage{iopams}
\usepackage{amsthm}

\newtheoremstyle{thm}
    {}
		{}
    {}
		{}
    {\bfseries}
		{.}
    { }
		{}
\theoremstyle{thm}

\newtheorem{thm}{Theorem}[section]
\newtheorem{prop}[thm]{Proposition}
\newtheorem{dfn}[thm]{Definition}
\newtheorem{rem}[thm]{Remark}
\newtheorem{lem}[thm]{Lemma}
\newtheorem*{expl}{Example}
\newtheorem{cor}[thm]{Corollary}

\newcommand{\divr}{\textnormal{div}\,}
\newcommand{\curl}{\textnormal{curl}\,}
\newcommand{\pl}{\parallel}
\newcommand{\ran}{\mathrm{ran}\,}

\begin{document}

\title[Approximate symmetries of guiding-centre motion]{Approximate symmetries of guiding-centre motion}

\author{J W Burby$^1$, N Kallinikos$^2$ and R S MacKay$^2$}

\address{$^1$ Los Alamos National Laboratory, Los Alamos, NM 87545, USA}
\address{$^2$ Mathematics Institute, University of Warwick, Coventry CV4 7AL, UK}
\ead{nikos.kallinikos@warwick.ac.uk}

\begin{abstract}
Quasisymmetry builds a third invariant for charged-particle motion besides energy and magnetic moment. We address quasisymmetry at the level of approximate symmetries of first-order guiding-centre motion. We find that the conditions to leading order are the same as for exact quasisymmetry if one insists that the symmetry is purely spatial. We also generalise to allow for approximate phase-space symmetries, and derive weaker conditions. The latter recover ``weak quasisymmetry" as a subcase, thus we prove it is spatial only to leading order, but also that it implies the existence of a wider class of independent approximate conserved quantities. Finally, we demonstrate that magnetohydrostatics imposes quasisymmetry to leading order.
\end{abstract}

\submitto{\jpa}

\section{Introduction}

Quasisymmetry was proposed \cite{B} as a way to achieve magnetic confinement and is the design principle \cite{NZ} underlying several modern optimised stellarators, including NCSX (partially constructed at PPPL) and HSX (built and operated at the University of Wisconsin-Madison). It is, in essence, a spatial symmetry of first-order guiding-centre motion that guarantees integrability.

In a previous work \cite{BKM}, necessary and sufficient conditions were derived for the existence of quasisymmetry, treating both system and symmetry as exact. It is worth noting that these hold for all nonzero values of charge, mass and magnetic moment.

Nevertheless, an approximate symmetry may be just as good as an exact one, especially since the guiding-centre system is only an approximation. Recently, it was suggested \cite{RHB} that approximate considerations of guiding-centre motion can relax the conditions of quasisymmetry.

In this paper, however, we show that any approximate spatial symmetry of the first-order guiding centre system must satisfy the quasisymmetry conditions \cite{BKM} to lowest order. This result contradicts the result presented in \cite{RHB}, which asserts that when a magnetic field is ``weakly quasisymmetric" there exists an approximate spatial symmetry for first-order guiding-centre motion.

Generalising from spatial symmetries to phase-space symmetries, we also find necessary and sufficient conditions for approximate symmetries that transform the parallel velocity of the guiding centre in addition to its position. In this way, we provide a set of weaker restrictions for an approximate conserved quantity. This set includes the case of \cite{RHB}, which proves to be a linearly parallel-velocity-dependent symmetry in first-order. We thereby confirm the approximate conserved quantity deduced in \cite{RHB}, even though the corresponding symmetry is not purely spatial. Moreover, we show that weak quasisymmetry implies a broader class of approximate conserved quantities than the single invariant considered in \cite{RHB}, which  we derive and generalise even further for genuine phase-space symmetries. Finally, we show that under the magnetohydrostatic assumption, approximate symmetries reduce to quasisymmetry, as well.

\section{Guiding-centre motion}
\label{gcsection}

The very notion of guiding centre is built on an approximate symmetry. It assumes that the motion of charged particles admits approximately a rotational symmetry about the magnetic field. As a result, the magnetic moment is an adiabatic invariant. This allows to reduce the original charged-particle motion to a 2-degree-of-freedom system for the gyrocentre, which tracks or, to put it the other way round, guides the particle. Guiding-centre motion averages over the fast, small-radius gyration, and describes the system reduced under gyrosymmetry. 

There have been various formulations of the guiding-centre system that agree to first order of approximation. Here we follow Littlejohn's, without taking into account electric fields, time-dependence or relativistic effects, which can be treated though accordingly.

Guiding-centre motion involves different features of the magnetic field that come into play both in contravariant and covariant components. This suggests the language of differential forms as more appropriate. Calculations and results support its use for brevity and hopefully clarity. That being said, notions and notation are kept to a minimum.\footnote{In short, the main tools are as follows. For any vector field $u$, $L_u$ denotes the Lie derivative with respect to $u$, $i_u$ stands for the interior product of a form with $u$, and $u^\flat$ the corresponding 1-form. Finally, $[u,w]=L_uw$ stands for the commutator of any two vector fields $u,w$. The only relations used next are limited to basic properties among $L_u$, $i_u$, the exterior derivative $d$ and the wedge product $\wedge$.} For the calculus of differential forms, besides classical textbooks we refer to the recent tutorial \cite{M} specifically adapted to 3D and plasma physics.

Throughout this paper we consider a 3-dimensional oriented smooth Riemannian manifold $Q$ equipped with associated volume-form $\Omega$, and assume that the magnetic field $B$ is nowhere zero on $Q$. We set $M=Q\times\mathbb{R}$, and also assume enough smoothness for all objects on $M$, wherever needed.

Let $x$ and $v_\pl$ denote the position and reduced velocity of the guiding centre, respectively. We think of $z=(x,v_\pl)$ as a point of $M$. Following \cite{L1}, the equations of first-order guiding-centre (FGC) motion for normalised constants ($m=q=1$) read
\begin{eqnarray}
\label{gc1}\eqalign{
\dot{x}&=\tilde{B}_\parallel^{-1}(v_\pl\tilde{B}+\epsilon\mu\,b\times\nabla|B|),\\
\dot{v}_\pl&=-\,\mu\tilde{B}_\parallel^{-1}\tilde{B}\cdot\nabla|B|,}
\end{eqnarray}
where $b=B/|B|$, $\tilde{B}=B+\epsilon v_\pl\curl b$ is the so called modified magnetic field, $\tilde{B}_\parallel=\tilde{B}\cdot b$, $\mu$ is the value of the magnetic moment, and $\epsilon$ is a scaling parameter that indicates the order of the guiding-centre approximation. For a weakly inhomogeneous magnetic field, $\epsilon\ll1$ says that the magnetic field varies slowly within a gyroradius $\rho$ and a gyroperiod $\tau$. This can be expressed as $\rho/L, \tau/T\propto\epsilon$, where $L$ and $T$ stand for the characteristic lengths and time (seen by the particle) over which $B$ changes appreciably. As both $\rho$ and $\tau$ are inversely proportional to the gyrofrequency $\Omega_B=q|B|/m$, one may adopt $\epsilon=m/q$ and treat $\mu$ as the magnetic moment per unit mass instead of normalisation.

An equivalent way to express the above system is
\begin{eqnarray}
\label{gc2}\eqalign{
\tilde{B}\times\dot{x}+\epsilon\dot{v}_\pl b+\epsilon\mu\nabla|B|=0,\\
b\cdot\dot{x}-v_\pl=0,}
\end{eqnarray}
explicitly defining $v_\pl$ as the component of the guiding centre velocity that is parallel to the magnetic field. Although (\ref{gc1}) is in solved form, (\ref{gc2}) is often more preferable to use and in fact precedes it in a Hamiltonian derivation.

In this form, the system admits a Hamiltonian formulation in the sense of $i_V\omega=-\,dH$ for $V=(\dot{x},\dot{v}_\pl)$, where the symplectic form $\omega$ and Hamiltonian function $H$ on $M$ (minus the set where $\tilde{B}_\pl=0$) are given by
\begin{eqnarray}
\label{symplectic}
\omega=\beta+\epsilon d(v_\pl b^\flat)\\
\label{hamiltonian}
H(x,v_\pl)=\epsilon(v_\pl^2/2+\mu|B|(x)).
\end{eqnarray}
Here $\beta=i_B\Omega$ is a 2-form on $M$ expressing the magnetic flux, and the projection from $M$ to $Q$ that pulls back $\beta$ and $b^\flat$ is dropped to simplify notation. Note that $d\beta=(\divr B)\Omega=0$, that is, $\beta$ is closed since $B$ is divergence-free.

In terms of the modified vector potential $\tilde{A}=A+\epsilon v_\pl b$, system (\ref{gc2}) is also formulated as a variational problem described by the Lagrangian \cite{L1}
\begin{equation}
\label{lagrangian}
L(x,v_\pl,\dot{x},\dot{v}_\pl)=\tilde{A}(x,v_\pl)\!\cdot\dot{x}-H(x,v_\pl),
\end{equation}
or equally the Poincar\'e-Cartan form $\alpha=Ldt=\tilde{A}^\flat-Hdt$ on extended state space \cite{L2}. The magnetic potential $A^\flat$ always exists locally, since by the Poincar\'e lemma the closed magnetic flux form $\beta$ is locally also exact on $Q$, i.e., $\beta=dA^\flat$. If $\int_S\beta=0$ for all surfaces $S$ representing the second homology group $H_2(Q)$, then $A^\flat$ is global.

\section{Approximate Hamiltonian symmetries}
\label{symmsection}

Approximate symmetries were introduced in \cite{BGI} in a framework very close to Lie's symmetry groups. An independent approach was presented in \cite{CG1} with a particular focus on dynamical systems and connections with normal forms. See also \cite{IK,CG2} for more reading. Here we adapt some of these notions to a Hamiltonian setup.

There are two things that a symmetry of an approximate system being approximate means. The first is that the symmetry as a transformation is approximate, and the second is that the symmetry leaves the system approximately invariant. For consistency, the order of approximation in both cases is the same as the system's. By symmetry in this paper, we will mean continuous symmetries on $M$.

The key ingredient to quantify approximate methods is that any object that depends on a small parameter $\epsilon$ is considered only up to terms $O(\epsilon^k)$ for some integer $k$. To apply this, it is useful to work on the equivalence class of functions described as follows. Any two functions $f,g$ that differ by $O(\epsilon^{k+1})$-terms are regarded as equal. To relax notation, we express this by writing
\begin{equation}
\label{appclass}
f(z;\epsilon)=g(z;\epsilon)+O(\epsilon^{k+1})~\Leftrightarrow~f(z;\epsilon)\approx g(z;\epsilon)
\end{equation}
for some fixed $k$.

Each equivalence class $[f]$ under $\approx$ has a natural representative, namely the $k$th-order Taylor polynomial of $f$ in $\epsilon$ about $\epsilon=0$. Thus, we can think of any $C^k$ function $f$ in $\epsilon$ defined on some manifold $M$ as
\begin{equation}
f(z;\epsilon)\approx\sum\limits_{i=0}^k\epsilon^if_i(z),
\end{equation}
$z\in M$. We do the same for any $\epsilon$-dependent differential form, vector field, and mapping whatsoever on $M$, assuming they are sufficiently smooth in a neighbourhood of $\epsilon=0$.

In the following we simply let $k=1$, as the forthcoming notions straightforwardly generalise to any order of approximation. So, the term ``approximate'' from now on will mean approximate of first order, unless stated otherwise.

\begin{dfn}
An approximate dynamical system on a manifold $M$ is a system of ordinary differential equations $\dot{z}=V(z;\epsilon)$ with $V\approx V_0+\epsilon V_1$, where $V_0,V_1$ are vector fields on $M$.
\end{dfn}

Under the equivalence $\approx$, note that any system that agrees up to first order with the vector field $V$ will do. We can express this by replacing $=$ with $\approx$ in (\ref{gc1}). Within this class, it is useful to work with Littlejohn's representative system that has a Hamiltonian structure.

We think of Hamiltonian systems in terms of symplectic forms, i.e., nondegenerate closed 2-forms on $M$. Relaxing the nondegeneracy requirement, presymplectic forms are just closed 2-forms. In the approximate setting, we ask for 

\begin{dfn}\label{Hamiltonian}
An approximate Hamiltonian system $(\omega,H)$ on a manifold $M$ is an approximate dynamical system $V$ that satisfies
\begin{equation}
\label{hamequation}
i_V\omega=-\,dH
\end{equation}
for $\omega=\omega_0+\epsilon\,\omega_1$, $H=H_0+\epsilon H_1$, where $\omega_0$ is a symplectic form, $\omega_1$ is a presymplectic form and $H_0,H_1$ are functions all on $M$.
\end{dfn}

We start by making precise the first aspect, what is an approximate transformation of a system.
\begin{dfn}
An approximate transformation on a manifold $M$ is a smooth map $\Phi:M\times I\longrightarrow M$ with $\Phi(z;\epsilon)\approx\Phi_0(z)+\epsilon\,\Phi_1(z)$, $z\in M$ that is invertible for each $\epsilon\in I$, where $I\subset\mathbb{R}$ is open and contains 0.
\end{dfn}

Continuous transformations means there is a family of transformations that depend continuously on a parameter in a manifold of dimension at least 1. Typically this family is required to form a group. In the approximate context, we have the following notion.

\begin{dfn}\label{dfn:appgroup}
A one-parameter approximate transformation group on a manifold $M$ is a set of approximate transformations $\Phi^\tau$ such that\vspace{-0.1cm}
\begin{enumerate}
\item $\Phi^\tau\approx\textnormal{Id}\textnormal{ iff } \tau=0$,
\item $\Phi^{\tau_1}\circ\Phi^{\tau_2}\approx\Phi^{\tau_1+\tau_2}$\vspace{-0.1cm}
\end{enumerate}
for all $\tau,\tau_1,\tau_2\in G$, where $G\subset\mathbb{R}$ is open and contains 0.
\end{dfn}

\begin{dfn}\label{dfn:appgen}
The infinitesimal generator of a one-parameter approximate transformation group $\Phi^\tau$ on a manifold $M$ is defined by
\begin{equation}
\label{eq:appgen}
U(z;\epsilon)\approx\left.\frac{d\Phi^\tau(z;\epsilon)}{d\tau}\right|_{\tau=0}
\end{equation}
\end{dfn}

The converse to this relation, which builds the group from the generator, is given by the solution $\tilde{z}=\Phi^\tau(z;\epsilon)$ to the initial-value problem $d\tilde{z}/d\tau\approx U(\tilde{z})$, $\tilde{z}(0)\approx z$. Equivalently, it can be constructed from the exponential map in the approximate sense, $\Phi^\tau\approx\exp(\tau U)$, where the exponential of a vector field is defined by following it for time one.

Combining Definitions \ref{dfn:appgroup} and \ref{dfn:appgen}, we see that the generator (\ref{eq:appgen}) is a vector field of the form $U=U_0+\epsilon U_1$, where $U_0=\left.d\Phi^\tau_0/d\tau\right|_{\tau=0}$ and $U_1=\left.d\Phi^\tau_1/d\tau\right|_{\tau=0}$.

Moving to the second point, a $k$th-order approximate transformation is a $k$th-order approximate symmetry of a $k$th-order approximate system if it leaves the system invariant up to $O(\epsilon^{k})$-terms. For an autonomous system described exactly by a vector field $V$, a necessary and sufficient condition for a vector field $U$ to be an exact symmetry is that $U$ and $V$ commute. In the approximate case, the symmetry criterion applies accordingly and is given below as a definition.

\begin{dfn}\label{symmetry}
An approximate symmetry of an approximate system $\dot{z}=V(z;\epsilon)$ on a manifold $M$ is a one-parameter approximate transformation group generated by a vector field $U$ on $M$ that satisfies $\left[U,V\right]\approx0$.
\end{dfn}

For approximate Hamiltonian systems and symmetries, the invariance criterion from the exact case also applies here accordingly and is given by the next definition.

\begin{dfn}\label{Hsymmetry}
An approximate Hamiltonian symmetry of an approximate Hamiltonian system $(\omega,H)$ on a manifold $M$ is a one-parameter approximate transformation group generated by a vector field $U$ on $M$ that satisfies $L_U\omega\approx0$ and $L_UH\approx0$.
\end{dfn}

For multiple future reference, it is worth noting that under $d\omega=0$ and $i_V\omega=-\,dH$, for any vector field $u$ we have the relations
\begin{eqnarray}
\label{Lomega}
L_U\omega&=di_U\omega\\
L_UH&=i_UdH=i_Vi_U\omega.
\label{Lhamiltonian}
\end{eqnarray}

Then an approximate version of Noether's theorem for Hamiltonian systems follows. While the map from constants of motion to Hamiltonian symmetries is automatic, its inverse though, given Definition \ref{Hsymmetry}, stumbles on the exactness of the closed 1-form $i_U\omega$. The next result offers a way out.

\begin{lem}
\label{homology}
Let $V$ be an approximate Hamiltonian system $(M,\omega,H)$ and $U$ an approximate Hamiltonian symmetry generator. If closed trajectories of the set of fields $fU+gV$ with $f,g$ arbitrary functions span $H_1(M)$, then $i_U\omega$ is approximately exact.
\end{lem}
\begin{proof}
If $U$ is a Hamiltonian symmetry generator, then $i_U\omega$ is closed up to first-order terms from (\ref{Lomega}), because $L_U\omega\approx0$. Also, $i_Xi_U\omega\approx0$ for any $X=fU+gV$ from (\ref{Lhamiltonian}), since $L_UH\approx0$. So, $\int_\gamma i_U\omega\approx0$, where $[\gamma]=X$, hence the result.
\end{proof}

\begin{dfn}
A function $K(z;\epsilon)=K_0(z)+\epsilon K_1(z)$, $z\in M$ is an approximate constant of motion for an approximate dynamical system $V$ on a manifold $M$ if $L_VK\approx0$.
\end{dfn}

\begin{thm}\label{noether}
If a function $K$ is an approximate constant of motion for the approximate Hamiltonian system $(\omega,H)$, then there exists an approximate Hamiltonian symmetry generated by a vector field $U$, unique up to equivalence, such that $i_U\omega\approx-\,dK$. Under the assumption of Lemma \ref{homology}, the converse is also true.
\end{thm}
\begin{proof}
For any function $K$, a vector field $U=U_0+\epsilon U_1$ such that $i_U\omega\approx-\,dK$ is well-defined, since $\omega_0$ is nondegenerate. This is because the zeroth-order terms $i_{U_0}\omega_0=-dK_0$ determine $U_0$ uniquely and then the first-order terms $i_{U_1}\omega_0+i_{U_0}\omega_1=-dK_1$ determine $U_1$ uniquely. Thus, we have $L_U\omega\approx0$ from (\ref{Lomega}), since $i_U\omega$ is closed up to first-order terms. If $L_VK\approx0$, then $L_UH\approx0$ too, because from (\ref{Lhamiltonian})
\begin{equation}
\label{HtoK}
L_UH=-\,i_Ui_V\omega=-\,i_VdK=-\,L_VK.
\end{equation}
In the other direction, if a vector field $U$ generates an approximate Hamiltonian symmetry, $i_U\omega\approx-\,dK$ for some global function $K$ by Lemma \ref{homology}. Then, using (\ref{HtoK}), $K$ is approximately conserved, because $L_UH\approx0$.
\end{proof}

\begin{rem}\label{varsymm}
Here as well as in \cite{BKM}, we have chosen to use Hamiltonian symmetries. Equally, one can address the same problem in terms of variational symmetries \cite{O}. In other words, assume that $U$ generates an approximate symmetry of the Lagrangian formulation for the system. This means that $U$ leaves $\int{\!L\,dt}$ invariant modulo boundary terms and up to $O(\epsilon)$-terms. Infinitesimally for GC motion it is expressed as $L_U\alpha\approx df$ for some arbitrary function $f(x,v_\pl)$, recalling $\alpha=\tilde{A}^\flat-Hdt$ from section \ref{gcsection}. This condition splits by $t$ into $L_U\tilde{A}^\flat\approx df$ and $L_UH\approx0$. Note that $d\tilde{A}^\flat=\omega$, so applying $d$ to the former gives $L_U\omega\approx0$. Therefore we recover Definition \ref{Hsymmetry}. The opposite direction requires Lemma \ref{homology}. Variational symmetries assume $K=U\cdot\tilde{A}-f$ is global from the beginning, and so Noether's formulation of Theorems \ref{noether} and \ref{prenoether} soon to follow does not require this lemma.
\end{rem}

\section{Noether's theorem for approximate presymplectic systems}
\label{casimirs}

The FGC system does not follow Definition \ref{Hamiltonian}, as $\omega_0=\beta$ is everywhere degenerate and therefore not symplectic. Consequently Theorem \ref{noether} does not apply in this case. Note that nondegeneracy of $\omega_0$ is actually a requirement only in the first direction of the theorem. Thus, if it fails then an approximate conserved quantity may correspond to more than one approximate Hamiltonian symmetry. The implications of this degeneracy for Noether's theorem are illustrated in this section.

A closed 2-form $\omega$ is called presymplectic. Thus, presymplectic forms may be degenerate and of variable rank. The rank of any 2-form $\omega$ is the dimension of the range of the associated linear map $\hat{\omega}$ from tangent vectors to cotangent vectors at each point, given by $\hat{\omega}(X)=i_X\omega$, and $\omega$ is degenerate if and only if the rank is less than the dimension of the manifold. For $\epsilon\neq0$ the guiding-centre form (\ref{symplectic}) in the exact scenario is symplectic except where $\tilde{B}_\parallel=0$. But for $\epsilon=0$ it reduces to $\beta$, which is closed ($\divr B=0$) and its rank is 2 everywhere (as $B=0$ is excluded) on the 4-dimensional manifold $M=Q\times\mathbb{R}$, so it is presymplectic and nowhere symplectic.

In general, the kernel of $\omega$ consists of all the vector fields that annihilate $\omega$ and degeneracy means nonzero $\ker\omega$ of dimension complementary to the range. In the approximate setup, in order to include any degeneracies arising from the equivalence relation $\approx$, we consider

\begin{dfn}
\label{kernel}
For a 2-form $\omega=\omega_0+\epsilon\,\omega_1$, $\ker\omega$ is the set of all approximate vector fields $S$ such that $i_S\omega\approx0$.
\end{dfn}

For a presymplectic form $\omega$, we continue to say a dynamical system $V$ is Hamiltonian if $i_V\omega=-\,dH$ for some $H$.  In contrast to the symplectic case, however, this does not have any solutions $V$ if $dH$ is not in the range of $\omega$, and if $dH$ is in the range then it has an affine space of solutions, consisting of one solution plus anything in its kernel, so $(\omega,H)$ no longer determines $V$ uniquely.  Thus, to specify a presymplectic Hamiltonian system we give $(V,\omega,H)$. We do the same for approximate systems, as well. In the sense of Definition \ref{kernel}, note that nondegeneracy of an approximate 2-form $\omega=\omega_0+\epsilon\,\omega_1$ requires only $\omega_0$ to be nondegenerate. Failing to meet this requirement, the guiding-centre form for $\epsilon\ll1$ can be said to be presymplectic. More generally, we say

\begin{dfn}
An approximate presymplectic Hamiltonian system $(V,\omega,H)$ on a manifold $M$ is an approximate dynamical system $V$ that satisfies $i_V\omega=-\,dH$ for $\omega=\omega_0+\epsilon\,\omega_1$, $H=H_0+\epsilon H_1$, where $\omega_0,\omega_1$ are presymplectic forms and $H_0,H_1$ are functions all on $M$, assuming $\omega_0$ is nowhere symplectic.
\end{dfn}


Next we present some symmetry aspects introduced by presympelctic forms. For approximate forms, the points we limit ourselves to are very similar to the exact case, thus we fine-tune them directly for approximate presymplectic systems. So, first, we adopt again Definition \ref{Hsymmetry} for Hamiltonian symmetries. Note though that the kernel of a presymplectic form gives rise automatically to Hamiltonian symmetries of all systems $V$ with the same presymplectic form $\omega$ regardless of the Hamiltonian function $H$.

\begin{prop}\label{trivial}
For an approximate presymplectic Hamiltonian system $(V,\omega,H)$, any vector field in $\ker\omega$ generates an approximate Hamiltonian symmetry for all $H$.
\end{prop}
\begin{proof}
Let $i_S\omega\approx0$ for some vector field $S$. Then $L_S\omega=di_S\omega\approx0$ from (\ref{Lomega}), and $L_SH=i_Vi_S\omega\approx0$ from (\ref{Lhamiltonian}), as $V$ respects the order of approximation of $i_S\omega$.
\end{proof}

Thus, such symmetries are not triggered by the dynamics of a particular system, they merely reduce it to a local symplectic submanifold. Moreover, they trivially satisfy the relation of Theorem \ref{noether} for a constant function. We say
\begin{dfn}
A trivial symmetry of an approximate presymplectic Hamiltonian system $(V,\omega,H)$ is a transformation generated by a vector field in $\ker\omega$.
\end{dfn}

\begin{rem}
Unlike the symplectic case, in the presymplectic case we cannot deduce that a Hamiltonian symmetry $U$ is a symmetry of $V$, only that $[U,V]\in\ker\omega$ and $[U+S,V]=0$ for some $S\in\ker\omega$.
\end{rem}

In order to restore the one-to-one correspondence in Noether's theorem, we need to consider equivalence classes of Hamiltonian symmetries, each differing from one another by a trivial one. This is where an approximate version meets a presymplectic one.

\begin{dfn}
\label{range}
For a 2-form $\omega=\omega_0+\epsilon\,\omega_1$, $\ran\omega$ is the set of all approximate 1-forms $i_X\omega$ for approximate vector fields $X$.
\end{dfn}

\begin{thm}\label{prenoether}
If a function $K$ is an approximate constant of motion for the approximate presymplectic Hamiltonian system $(V,\omega,H)$ with $dK\in\textnormal{ran}\,\omega$, then there exists an approximate Hamiltonian symmetry generated by any vector field $U+S$ such that $i_U\omega\approx-\,dK$ and $S\in\ker\omega$. Under the assumption of Lemma \ref{homology}, the converse is also true.
\end{thm}
\begin{proof}
The proof follows from Proposition \ref{trivial} and along the same lines as the proof of Theorem \ref{noether}. In the first direction, for any function $K$ with $dK\in\ran\omega$, a vector field $\tilde{U}$ such that $i_{\tilde{U}}\omega\approx-\,dK$ can be defined uniquely modulo elements of $\ker\omega$, i.e., $\tilde{U}=U+S$, where $i_U\omega\approx-\,dK$. Then, as in Theorem \ref{noether}, $\tilde{U}$ generates an approximate Hamiltonian symmetry. In the other direction, if $U+S$ generates an approximate Hamiltonian symmetry, then so does $U$ by Proposition \ref{trivial}. Then, as in Theorem \ref{noether}, $i_U\omega\approx-\,dK$ by Lemma \ref{homology} and $K$ is an approximate conserved quantity.
\end{proof}

For more general presymplectic systems, where $V$ is not unique or only exists on a submanifold of $M$, see \cite{FP,LD} for a (purely) presymplectic version of Noether's theorem.

\subsection{The guiding-centre case}
Back to FGC motion,
\begin{prop}
\label{gcrange}
The range of the guiding-centre 2-form $\omega=\beta+\epsilon d(v_\pl b^\flat)$ consists of all the 1-forms $a=a_0+\epsilon a_1$ on $M$ such that $a_0\in\ran\beta$.
\end{prop}
\begin{proof}
Let $a=a_0+\epsilon a_1$ be a 1-form such that $i_X\omega\approx a$ for some vector field $X=X_0+\epsilon X_1$. Then $i_{X_0}\beta=a_0$ and $i_{X_1}\beta+i_{X_0}d(v_\pl b^\flat)=a_1$. Using $d(v_\pl b^\flat)=dv_\pl\wedge b^\flat+v_\pl db^\flat$, the second equation gives
\begin{equation}
\label{ran2nd}
i_{X_1}\beta+(i_{X_0}dv_\pl)b^\flat-(i_{X_0}b^\flat)dv_\pl+v_\pl i_{X_0}db^\flat=a_1.
\end{equation}
The first three terms show that $a_1$ is any 1-form for arbitrary $X_1$ and $v_\pl$-, $b$-components of $X_0$. The latter do not enter the first condition, hence the result.
\end{proof}

\begin{prop}
\label{gckernel}
The kernel of the guiding-centre 2-form $\omega=\beta+\epsilon d(v_\pl b^\flat)$ consists of all the vector fields $S=\epsilon S_1$ on $M$ such that $S_1\in\ker\beta$.
\end{prop}
\begin{proof}
Let $S=S_0+\epsilon S_1$ be a vector field such that $i_S\omega\approx0$. Then $i_{S_0}\beta=0$ and $i_{S_1}\beta+i_{S_0}d(v_\pl b^\flat)=0$. Now, the second equation gives
\begin{equation}
\label{ker2nd}
i_{S_1}\beta+(i_{S_0}dv_\pl)b^\flat-(i_{S_0}b^\flat)dv_\pl+v_\pl i_{S_0}db^\flat=0.
\end{equation}
Contracting the above with $\partial_{v_\pl}$, we have $i_{S_0}b^\flat=0$. The latter together with the first equation yields $i_{S_0}\Omega=0$, applying $i_{S_0}$ on $\beta\wedge b^\flat=|B|\Omega$. Then the last term in (\ref{ker2nd}) also vanishes, because $i_{S_0}db^\flat=i_{S_0}i_c\Omega=-\,i_ci_{S_0}\Omega=0$, where $c=\curl b$. So, if we contract (\ref{ker2nd}) with $b$, we get $i_{S_0}dv_\pl=0$. Thus, $S_0=0$ because $i_{S_0}\Omega=0$ and $i_{S_0}dv_\pl=0$. Then (\ref{ker2nd}) reduces to just $i_{S_1}\beta=0$.
\end{proof}

\begin{cor}
\label{gcnoether}
For FGC motion, there is a one-to-one correspondence between approximate constants of motion $K$ with $K_0=\textnormal{const.}$ being flux surfaces and classes of approximate Hamiltonian symmetries $U+\epsilon(fb,g)$ where $f,g$ are any functions, and $i_U\omega\approx-\,dK$.
\end{cor}

\begin{proof}
The range of $\beta$ on $M$ consists of all the 1-forms on $Q$ that vanish on $B$, and so for exact 1-forms $dK_0$ this means $i_BdK_0=0$, i.e., $K_0=\textnormal{const.}$ is a flux surface. The kernel of $\beta$ on $M$ consists of all the vector fields $(fb,g)$, where $f,g$ are arbitrary functions. The result follows from Theorem \ref{prenoether} and Propositions \ref{gcrange}, \ref{gckernel}.
\end{proof}

\begin{rem}
Note that for all values of $\mu$ the vector field $V_0=(v_\pl b,-\mu\,b\cdot\nabla|B|)$ spans $\ker\beta$, assuming $b\cdot\nabla|B|\neq0$. Then, $i_{V_0}dK_0=0$ says that $dK_0$ belongs to $\ran\omega$ automatically. In other words, instead of asking $K_0=\textnormal{const.}$ to be a flux surface in Corollary \ref{gcnoether} we can ask $K_0$ to be independent of $\mu$ when $b\cdot\nabla|B|\neq0$.
\end{rem}




\section{Approximate quasisymmetry}

In this section, we address approximate Hamiltonian spatial symmetries for guiding-centre motion. Our goal is to see how quasisymmetry can be approximated using the guiding-centre approximation. Either in the exact or the approximate framework,


\begin{dfn}
Quasisymmetry is a Hamiltonian symmetry on $Q$ of FGC motion for all values of the magnetic moment.
\end{dfn}

\begin{thm}\label{appqs}
Given a magnetic field $B$, a vector field $u=u_0+\epsilon u_1$ on $Q$ generates an approximate quasisymmetry if and only if $L_{u_0}\beta=0$, $L_{u_0}b^\flat=0$, $L_{u_0}|B|=0$, $L_{u_1}\beta=0$.
\end{thm}
\begin{proof}
Substitute $u$, $\omega$ and $H$ into the conditions of Definition \ref{Hsymmetry} and split up by different powers of $\epsilon$, dropping any second-order terms. Starting with $L_uH\approx0$, we get $L_{u_0}|B|=0$. Similarly from $L_u\omega\approx0$, we have $L_{u_0}\beta=0$ and
\begin{equation}
\label{p1}
L_{u_0}d(v_\pl b^\flat)+L_{u_1}\beta=0
\end{equation}
from the zero- and first-order terms, respectively. Now
\begin{equation}
\label{p2}
L_{u_0}d(v_\pl b^\flat)=dL_{u_0}(v_\pl b^\flat)=d(v_\pl L_{u_0}b^\flat)=dv_\pl\wedge L_{u_0}b^\flat+v_\pl dL_{u_0}b^\flat.
\end{equation}
Thus, contracting (\ref{p1}) with $\partial_{v_\pl}$, we get $L_{u_0}b^\flat=0$. Substituting this into (\ref{p2}) gives $L_{u_0}d(v_\pl b^\flat)=0$, and so (\ref{p1}) yields $L_{u_1}\beta=0$. Going in the opposite direction, it is straightforward to see that the converse is also true.
\end{proof}

As shown in \cite{BKM}, under the above conditions $u_0$ satisfies several additional properties such as $\divr u_0=0$, $[u_0,B]=0$, $[u_0,J]=0$, $L_{u_0}(u_0\cdot b)=0$ and others. Note that $u_0$ and $u_1$ are uncoupled.

\section{Approximate $v_\pl$-symmetries}
\label{sc:appvs}

Subsequently we ask how departures from quasisymmetry that depend on parallel velocity can relax the conditions on $B$ for FGC motion to have a symmetry. Thus, we investigate the conditions for an approximate Hamiltonian symmetry on phase space $M$, which will often be referred to simply as approximate symmetries. As it turns out (Theorem \ref{appvs}), symmetry generators for all values of the magnetic moment have zero component in the parallel-velocity direction. Thus, we will also refer to symmetries on $M$ as $v_\pl$-symmetries, which is short for parallel-velocity-dependent symmetries generated on $Q$.

Symmetries that involve velocities are not new to charged particle motion. Gyrosymmetry is an example of an approximate Hamiltonian symmetry involving the perpendicular velocity to leading order.

\begin{expl}\normalfont
Consider the full particle's motion on the cotangent bundle $T^*Q$ with symplectic form $\omega=\beta+dv\wedge dx$, where $v$ is the particle's velocity. In the case of a homogeneous magnetic field, the magnetic moment $\mu=v_\perp^2/2|B|$ is an exact constant of motion and corresponds via $i_U\omega=-\,d\mu$ to the exact symmetry generated by $U=(v_\perp/|B|,v_\perp\!\times b)$ on $T^*Q$, where $v_\perp$ is the perpendicular velocity vector of the particle.

For a weakly-inhomogeneous $B$, we have $\omega=\beta+\epsilon\,dv\wedge dx$ 
for $\epsilon\ll1$. 
The magnetic moment now extends to an adiabatic invariant $K=\epsilon^2 v_\perp^2/2|B|+O(\epsilon^3)$, which under $i_U\omega=-\,dK$ corresponds to an approximate Hamiltonian symmetry that extends to all orders, generated by the vector field $U=U_0+\epsilon U_1+O(\epsilon^2)$ with
\begin{eqnarray}
\label{gs1}
\fl U_0&=(0,v_\perp\!\times b),\\
\label{gs2}
\fl U_1&=|B|^{-1}\!\left(v_\perp,\left\{\!\left(v_\pl\,b\cdot c-2K_1\,n\cdot\nabla|B|\right)\!b-v_\pl c+\left[b\times(v_\perp\cdot\nabla b)+n\cdot\nabla b\right]\!/2\right\}\!\times v\right)\!,
\end{eqnarray}
where $n=b\times v_\perp$ and $c=\curl b$, as shown in \ref{GSappendix}. Thus, the exact symmetry from the homogeneous case splits between terms of different order. Note that the leading order of $U$ is less by two than $K$'s, same as with $V$ and $H$. 
Formulas (\ref{gs1})-(\ref{gs2}) recover the one in \cite{JS}, where a coordinate-free way is presented to build the so called roto-rate as a means to gyrosymmetry and the corresponding adiabatic invariant to all orders for nearly-periodic systems.
\end{expl}

For considerations of general symmetries, we need to work on $M=Q\times\mathbb{R}$ with volume form $\Omega\wedge dv_\pl$. For any vector field $U$ on $M$, we denote by $u$ the projection of $U$ on $Q$, i.e., the spatial components of $U$ collectively, and by $w$ the component of $U$ in the parallel-velocity direction; we write $U=(u,w)$.


\begin{thm}
\label{appvs}
Given a magnetic field $B$, a vector field $U=(u,w)=(u_0+\epsilon u_1,w_0+\epsilon w_1)$ on $M$ generates an approximate Hamiltonian symmetry of FGC motion if and only if
\begin{eqnarray}
\label{vs1}L_{u_0}\beta=0,\\
\label{vs2}d(v_\pl L_{u_0}b^\flat)+L_{u_1}\beta=0,\\
\label{vs3}L_{u_0}|B|=0,\\
w=0.
\end{eqnarray}
\end{thm}
\begin{proof}
From $L_UH\approx0$, we have $v_\pl w_0+\mu L_{u_0}|B|=0$. Then for all values of $\mu$, we get (\ref{vs3}) and $w_0=0$. Given Corollary \ref{gcnoether}, take also $w_1=0$ under the equivalence by trivial symmetries. Thus, $U$ has overall zero velocity-component $w$.

Since $w=0$, $L_U\omega\approx0$ reduces to $L_u\omega\approx0$. From the latter, we obtain (\ref{vs1}) from the zeroth-order terms, and $L_{u_0}d(v_\pl b^\flat)+L_{u_1}\beta=0$ from the first-order ones. This in turn gives (\ref{vs2}), because $L_{u_0}d(v_\pl b^\flat)=d(v_\pl L_{u_0}b^\flat)$.
\end{proof}

Next we explore some further consequences and also subcases.

\begin{thm}
\label{appvscons}
If $u=u_0+\epsilon u_1$ generates an approximate Hamiltonian $v_\pl$-symmetry of FGC motion, then $\divr u_0=0$, $[u_0,B]=0$ and $i_bL_{u_0}b^\flat=0$. Furthermore,\vspace{-0.1cm}
\begin{enumerate}
	\item If $u_1=0$, then $[u_0,c]=0$ and $i_cL_{u_0}b^\flat=0$, where $c=\curl b$.
	\item If $u_0$ is spatial, then $L_{u_0}b^\flat=i_Bi_{\partial_{v_\pl}\!u_1}\Omega$.
\end{enumerate}
\end{thm}
\begin{proof}
Note first that $L_{u_0}\Omega\wedge dv_\pl=di_{u_0}\Omega\wedge dv_\pl=(\divr u_0)\Omega\wedge dv_\pl$. Take then the first symmetry condition (\ref{vs1}) and split by spatial and velocity components. In order to do this, wedge with $dv_\pl$ and contract with $\partial_{v_\pl}$, respectively. In the first case, write $L_{u_0}\beta=i_{[u_0,B]}\Omega+i_BL_{u_0}\Omega$, and therefore $L_{u_0}\beta\wedge dv_\pl=(i_{[u_0,B]}+\divr u_0\,i_B)\Omega\wedge dv_\pl$, where $\Omega\wedge dv_\pl$ is nondegenerate. In the second case, we have $i_{\partial_{v_\pl}}\!L_{u_0}\beta=i_{\partial_{v_\pl}\!u_0}i_B\Omega$, since $[\partial_{v_\pl},u_0]=\partial_{v_\pl}\!u_0$. Thus, the first symmetry condition splits into
\begin{eqnarray}
\label{con11}
[u_0,B]+(\divr u_0)B=0,\\
\label{con12}
i_{\partial_{v_\pl}\!u_0}i_B\Omega=0.
\end{eqnarray}

In the same way, split the second symmetry condition (\ref{vs2}). First of all, note that $db^\flat=i_c\Omega$ and write $d(v_\pl L_{u_0}b^\flat)=dv_\pl\wedge L_{u_0}b^\flat+v_\pl L_{u_0}db^\flat$. Similarly then wedge with $dv_\pl$, using now $L_{u_0}db^\flat=i_{[u_0,c]}\Omega+i_cL_{u_0}\Omega$, as well. In contracting with $\partial_{v_\pl}$, note that $i_{\partial_{v_\pl}}\!L_{u_0}b^\flat=i_{\partial_{v_\pl}\!u_0}b^\flat$ and $i_{\partial_{v_\pl}}\!L_{u_0}db^\flat=i_{\partial_{v_\pl}\!u_0}i_c\Omega$, since $b^\flat$ lies on $Q$. 
Thus, as before, the second condition gives
\begin{eqnarray}
\label{con21}
v_\pl[u_0,c]+v_\pl(\divr u_0)c+[u_1,B]+(\divr u_1)B=0,\\
\label{con22}
L_{u_0}b^\flat-(i_{\partial_{v_\pl}\!u_0}b^\flat)dv_\pl+v_\pl i_{\partial_{v_\pl}\!u_0}i_c\Omega+i_{\partial_{v_\pl}\!u_1}i_B\Omega=0.
\end{eqnarray}

From (\ref{con22}) and (\ref{con12}), we have $i_bL_{u_0}b^\flat=0$.

Now using this, $L_{u_0}|B|=i_{[u_0,B]}b^\flat$. Then the third symmetry condition (\ref{vs3}) combined with (\ref{con11}) gives $\divr u_0=0$ and so $[u_0,B]=0$, as well.

The remaining two statements are automatic from (\ref{con21}), substituting $\divr u_0=0$, and (\ref{con22}).
\end{proof}


\begin{rem}\label{appvsconseq}
From the proof of Theorem \ref{appvscons}, we see that under $L_{u_0}\beta=0$, $L_{u_0}|B|=0$, the condition $\divr u_0=0$ is equivalent to either $[u_0,B]=0$ or $i_bL_{u_0}b^\flat=0$. 
\end{rem}

Note that, in spite of $w=0$, the first two symmetry conditions (\ref{vs1})-(\ref{vs2}) still have $v_\pl$-components. The next result shows that (\ref{vs1}) can be reduced to a condition on $Q$ and gives a reformulation of (\ref{vs2}).

\begin{lem}
\label{appvslem}
For any two $v_\pl$-dependent vector fields $u_0,u_1$ on $Q$, the conditions (\ref{vs1}) and (\ref{vs2}) hold if and only if
\begin{eqnarray}
\label{vs11}
i_{u_0}i_B\Omega=d\psi_0,\\
\label{vs21}
v_\pl L_{u_0}b^\flat+i_{u_1}i_B\Omega=d\psi_1,
\end{eqnarray}
respectively, where $\psi_0$ is a spatial function on $Q$ and $\psi_1$ is a function on $M$, both defined at least locally. Under (\ref{vs11})-(\ref{vs21}), $\psi_1$ is spatial if and only if $u_0$ is.
\end{lem}
\begin{proof}
$L_{u_0}\beta=di_{u_0}\beta$, since $\beta$ is closed. Thus, by the Poincar\'e lemma, (\ref{vs1}) holds if $i_{u_0}\beta=d\psi_0$ for some local function $\psi_0$ on $M$. The $v_\pl$-component then gives $\partial_{v_\pl}\!\psi_0=0$, since $i_{u_0}\beta$ is a 1-form on $Q$. Similarly, (\ref{vs2}) holds if $v_\pl L_{u_0}b^\flat+i_{u_1}\beta=d\psi_1$ for some local function $\psi_1$ on $M$. In the other direction, $\psi_0$ and $\psi_1$ can be global.

Now, on the one hand, the $v_\pl$-derivative of (\ref{vs11}) gives $i_{\partial_{v_\pl}\!u_0}i_B\Omega=0$, since $B$ and $\psi_0$ are spatial. On the other, the $v_\pl$-component of (\ref{vs21}) yields $v_\pl i_{\partial_{v_\pl}\!u_0}b^\flat=\partial_{v_\pl}\!\psi_1$. To see this, use $L_{u_0}b^\flat=i_{u_0}db^\flat+di_{u_0}b^\flat$ and note that $i_{u_1}i_B\Omega$ lies on $Q$, and so does $i_{u_0}db^\flat$. Therefore, when both conditions hold, we have $\partial_{v_\pl}\!u_0\times B=0$ and $v_\pl\,\partial_{v_\pl}\!u_0\cdot b=\partial_{v_\pl}\!\psi_1$. Hence $\partial_{v_\pl}\!\psi_1=0$ if and only if $\partial_{v_\pl}\!u_0=0$.
\end{proof}

Included in this section to treat the general case, the above lemma can be combined with either Theorem \ref{appqs} or \ref{appvs}.

From Theorem \ref{appvs}, we see already that for a general approximate (Hamiltonian) symmetry, $u_0$ and $u_1$ are now related via (\ref{vs2}). From Lemma \ref{appvslem} and (\ref{vs21}) in particular, we can express $u_1$ in terms of $u_0$ and give another characterisation of approximate Hamiltonian symmetries.

\begin{thm}\label{appvs2}
Given a magnetic field $B$, a vector field $U=(u,w)=(u_0+\epsilon u_1,w_0+\epsilon w_1)$ on $M$ generates an approximate Hamiltonian symmetry of FGC motion up to trivial symmetries if and only if $L_{u_0}\beta=0$, $L_{u_0}|B|=0$, $w=0$, and
\begin{equation}
\label{u1}
u_1=|B|^{-1}b\times(v_\pl X_0-\nabla\psi_1),
\end{equation}
where $X_0=\curl b\times u_0 + \nabla(u_0\cdot b)$, $\nabla$ denotes the spatial gradient and $\psi_1$ is a flux function on $M$ defined at least locally such that
\begin{equation}
\label{u1c2}
\partial_{v_\pl}\!\psi_1=v_\pl\,b\cdot\partial_{v_\pl}\! u_0.
\end{equation}
\end{thm}

\begin{proof}
First of all, note that $X_0^\flat\wedge dv_\pl=L_{u_0}b^\flat\wedge dv_\pl$. Thus, the spatial part of (\ref{vs21}) is $v_\pl X_0+B\times u_1=\nabla\psi_1$. Cross then with $b$ and drop any trivial symmetries to arrive at (\ref{u1}). Dotting with $b$ yields $b\cdot\nabla\psi_1=v_\pl\,b\cdot X_0=i_bL_{u_0}b^\flat=0$ by Theorem \ref{appvscons}. The velocity part of (\ref{vs21}), as shown in the proof of Lemma \ref{appvslem}, gives (\ref{u1c2}).\end{proof}

Note that equally we can replace $X_0$ with $|B|^{-1}(J\times u_0 + \nabla(u_0\cdot B))$ in (\ref{u1}). Condition (\ref{u1c2}) says that the $v_\pl$-dependence of $\psi_1$ is determined by the $v_\pl$-dependence of $u_0$. For example, if $u_0$ is an $n$th-order polynomial in $v_\pl$, then so is $\psi_1$. 

\begin{rem}
From (\ref{u1}), we deduce that, since $u_0$ is spatial if and only if $\psi_1$ is spatial, $u_1$ is nonzero up to trivial symmetries unless $u_0$ depends on $v_\pl$ or $L_{u_0}b^\flat=0$. In other words, we cannot have both spatial $u_0$ and zero $u_1$, assuming $L_{u_0}b^\flat\neq0$.
\end{rem}

To connect with other formulations, we express some key relations of the previous results in vector calculus notation in \ref{VCappendix}.

\section{Approximate flux surfaces and constants of motion}

Back to Lemma \ref{appvslem}, we see that $B\cdot\nabla\psi_0=0$ from (\ref{vs11}). Thus, even in the case of an approximate phase-space (Hamiltonian) symmetry there exists a flux function $\psi_0$, at least locally, and we assume it is global.

From (\ref{vs21}) and $i_bL_{u_0}b^\flat=0$ we also have $B\cdot\nabla\psi_1=0$, as stated in Theorem \ref{appvscons}. Thus, there exists an approximate, generalised notion of a flux function given by
\begin{equation}
\psi=\psi_0+\epsilon\psi_1
\end{equation}
assuming $\psi_1$ is also global. We say generalised, because although $\psi_0$ is spatial, $\psi_1$ may depend on the parallel velocity. 

From now on, we will assume that both $\psi_0$ and $\psi_1$ are global, and we will refer to the level sets of $\psi_0$ as flux surfaces.


Finally, to construct the approximate conserved quantity $K$ that arises from an approximate Hamiltonian symmetry $U$ in general, we employ Corollary \ref{gcnoether}. Recall that $K$ is uniquely determined by $U$ via $i_U\omega\approx-\,dK$ and vice versa, since trivial symmetries have been factored out. For any vector field $U=U_0+\epsilon U_1=(u_0,w_0)+\epsilon(u_1,w_1)$, we have
\begin{eqnarray*}
i_U\omega&\approx i_{U_0}\beta+\epsilon[i_{U_0}d(v_\pl b^\flat)+i_{U_1}\beta]=i_{u_0}i_B\Omega+\epsilon[(L_{U_0}-di_{u_0})v_\pl b^\flat+i_{u_1}i_B\Omega]\\
&=i_{u_0}i_B\Omega+\epsilon[w_0b^\flat+v_\pl L_{u_0}b^\flat+i_{u_1}i_B\Omega-d(v_\pl u_0\cdot b)],
\end{eqnarray*}
and so, using Theorem \ref{appvs} and Lemma \ref{appvslem}, we arrive at
\begin{equation}
\label{invariant}
K=-\,\psi_0-\epsilon(\psi_1-v_\pl u_0\cdot b).
\end{equation}

This is a generalisation of the exact invariant \cite{BKM} in that it introduces the generalised flux function $\psi_1$. Note that this formula applies for spatial symmetries too, only the conditions on $u_0$ change. By Lemma \ref{appvslem}, however, the function $K$ is linear in the velocity if $u_0$ is spatial, but nonlinear otherwise. Interestingly enough, $u_1$ does not enter explicitly.

\section{Weak quasisymmetry}

Given Theorems \ref{appvscons}, \ref{appvs2} and Remark \ref{appvsconseq}, we conclude that an approximate Hamiltonian $v_\pl$-symmetry generator $u_0+\epsilon u_1$ satisfies the conditions
\begin{equation}
\label{vs0}
L_{u_0}\beta=0, \divr u_0=0, L_{u_0}|B|=0
\end{equation}
to zero order and the first-order term $u_1$ is given by (\ref{u1}). The only additional condition (\ref{u1c2}) restricts the velocity-dependence between $u_0$ and $\psi_1$, and so it is automatic if either one is spatial by Lemma \ref{appvslem}.

Here we address the converse with no assumption on $u_1$ whatsoever. Leaving (\ref{u1c2}) aside, we may assume that $\psi_1$ (and so $u_0$) is independent of $v_\pl$. To connect also with \cite{RHB}, we first treat the restriction to $\psi_1=0$ considered there.

\begin{prop}
\label{RHBi}
If $i_{u_0}i_B\Omega=d\psi_0$, $\divr u_0=0$ and $L_{u_0}|B|=0$ with spatial $u_0,\psi_0$, then $p=-\,\psi_0+\epsilon\,v_\pl u_0\cdot b$ is an approximate conserved quantity for FGC motion.
\end{prop}
\begin{proof}
To check the approximate invariance of $p$, compute $L_Vp$ with $V=(\dot{x},\dot{v}_\pl)$ up to $O(\epsilon)$-terms. Thus, assuming $u_0$ is independent of $v_\pl$, we have
\begin{eqnarray}
\label{Pinvariance}
L_Vp&=-L_V\psi_0+\epsilon\,u_0\cdot b\,L_Vv_\pl+\epsilon\,v_\pl L_V(u_0\cdot b)\nonumber\\
&=-\,i_{\dot{x}}d\psi_0+\epsilon\,\dot{v}_\pl u_0\cdot b+\epsilon\,v_\pl i_{\dot{x}}d(u_0\cdot b)\nonumber\\
&=-\,i_{\dot{x}}i_{u_0}i_B\Omega+\epsilon\,\dot{v}_\pl u_0\cdot b+\epsilon\,v_\pl i_{\dot{x}}(L_{u_0}b^\flat-i_{u_0}db^\flat)\nonumber\\
&=-\,i_{\dot{x}}i_{u_0}i_{\tilde{B}}\Omega+\epsilon\,\dot{v}_\pl u_0\cdot b+\epsilon\,v_\pl i_{\dot{x}}L_{u_0}b^\flat\nonumber\\
&\approx-\,\epsilon\,i_{u_0}(\dot{v}_\pl b^\flat+\mu d|B|)+\epsilon\,\dot{v}_\pl u_0\cdot b+\epsilon\,v_\pl i_bL_{u_0}b^\flat\approx0,
\end{eqnarray}
using (\ref{gc1})-(\ref{gc2}) in the penultimate equality and Remark \ref{appvsconseq} in the last one.
\end{proof}

As with the general case of $K$, the approximate constant $p$ does not consider $u_1$. One might ask if $u_0$ on $Q$ is the corresponding Hamiltonian symmetry generator under (\ref{vs0}). Theorem \ref{appqs} rules out this possibility. One can verify that $u_0$ does not even generate an approximate symmetry of the guiding-centre equations themselves, regardless of the Hamiltonian structure, in the sense of Definition \ref{symmetry} (or at least modulo $\ker\omega$ to include any degeneracies). However, in light of Noether's theorem adapted here successively, leading to Corollary \ref{gcnoether}, we can construct the symmetry from $p$. Either by direct calculation or section \ref{sc:appvs} going backwards, we obtain

\begin{prop}
\label{RHBs}
The approximate conserved quantity $p=-\,\psi_0+\epsilon\,v_\pl u_0\cdot b$ with spatial $u_0,\psi_0$ satisfying $i_{u_0}i_B\Omega=d\psi_0$, $\divr u_0=0$, $L_{u_0}|B|=0$ corresponds to the approximate Hamiltonian symmetry generated by $u_p=u_0+\epsilon v_\pl|B|^{-1}b\times X_0$, where $X_0=\curl b\times u_0 + \nabla(u_0\cdot b)$, up to trivial symmetries.
\end{prop}

Proposition \ref{RHBi} agrees with \cite{RHB} that $p$ is an approximate conserved quantity under conditions (\ref{vs0}). Contrary to \cite{RHB}, however, Proposition \ref{RHBs} shows that under these conditions the arising symmetry is not purely spatial, but it is spatial to lowest order and depends linearly on parallel velocity in first order.

This is only an example of such symmetries; within this symmetry class we could have in general a nonzero, spatial $\psi_1$. Although they escape quasisymmetry, these symmetries are a weak version of it.

\begin{dfn}
\label{dfn:weakvs}
A weak quasisymmetry is an approximate Hamiltonian symmetry of FGC motion which is spatial to leading order and nontrivially linear in $v_\pl$ to first order.
\end{dfn}

Propositions \ref{RHBi}-\ref{RHBs} indicate that a spatial vector field $u_0$ that satisfies (\ref{vs0}) is the zeroth-order term of a weak quasisymmetry generator. We extend this to include the case of spatial $\psi_1\neq0$.

\begin{thm}
\label{thm:weakvs}
Assume $u_0$ is a vector field on $Q$ and $u_1=|B|^{-1}b\times(v_\pl X_0-\nabla\psi_1)$ with $X_0=\curl b\times u_0 + \nabla(u_0\cdot b)\neq0$ and $\psi_1$ a flux function on $Q$. The following are equivalent.\vspace{-0.1cm}
\begin{enumerate}
	\item $u=u_0+\epsilon u_1$ generates a weak quasisymmetry;
	\item $i_{u_0}i_B\Omega=d\psi_0$, $\divr u_0=0$ and $L_{u_0}|B|=0$.
\end{enumerate}
\end{thm}
\begin{proof}
If $u$ generates a weak quasisymmetry then from Theorems \ref{appvs}-\ref{appvscons} and Lemma \ref{appvslem} we see that the conditions (ii) hold.

In the opposite direction, note first that since $u_0$ is spatial, $L_{u_0}b^\flat$ lies on $Q$ and reduces to $L_{u_0}b^\flat=X_0^\flat$ (Theorem \ref{appvs2}). Now, under the other two conditions, $\divr u_0=0$ is equivalent to $b\cdot X_0=0$ (Remark \ref{appvsconseq}). Together with $B\cdot\nabla\psi_1=0$, they guarantee that we can define a vector field $u_1$ from (\ref{u1}). Then, Theorem \ref{appvs2} says that $u_0+\epsilon u_1$ is a Hamiltonian $v_\pl$-symmetry with (\ref{u1c2}) trivially satisfied. Since $X_0$ and $\psi_1$ are independent of $v_\pl$, it is a weak quasisymmetry.
\end{proof}

\section{Approximate $\mu$-symmetries}

We may as well enlarge the set of symmetries by allowing them to depend on $\mu$, and look for Hamiltonian symmetries on $M$ for specific values of the magnetic moment. However, the next theorem shows that these reduce to phase-space Hamiltonian symmetries.

\begin{thm}
For FGC motion, every approximate $\mu$-dependent Hamiltonian symmetry on $M$ is an approximate Hamiltonian symmetry up to trivial symmetries.
\end{thm}

\begin{proof}
Let $U(x,v_\pl,\mu,\epsilon)=U_0(x,v_\pl,\mu)+\epsilon U_1(x,v_\pl,\mu)$ be the symmetry generator on $M$ with $U_i=(u_i,w_i)$, $i=1,2$, where $w_1=0$ up to trivial symmetries. We work our way partly through Theorems \ref{appvs} and \ref{appvscons} and modify them suitably.

From $L_U\omega\approx0$, the zeroth-order terms give $L_{U_0}\beta=0$, which reduces again to (\ref{vs1}), since $\beta$ is a spatial form on $Q$. But the first-order terms yield $L_{U_0}d(v_\pl b^\flat)+L_{U_1}\beta=0$, and note that $L_{U_0}(v_\pl b^\flat)=w_0 b^\flat+v_\pl L_{u_0}b^\flat$, as $b^\flat$ is spatial too. So now, instead of (\ref{vs2}) the second symmetry condition reads
\begin{equation}
\label{mu1}
d(v_\pl L_{u_0}b^\flat+w_0 b^\flat)+L_{u_1}\beta=0.
\end{equation}

Take then the $\mu$-component of (\ref{vs1}) and (\ref{mu1}), i.e., contract them with $\partial_\mu$. Similarly to Theorem \ref{appvscons}, we find
\begin{eqnarray}
\label{mu2}
i_{\partial_\mu u_0}i_B\Omega=0,\\
\label{mu3}
-\,(i_{\partial_\mu u_0}b^\flat)dv_\pl+v_\pl i_{\partial_\mu u_0}i_c\Omega+(\partial_\mu w_0)b^\flat+i_{\partial_\mu u_1}i_B\Omega=0,
\end{eqnarray}
respectively. The $v_\pl$-component of (\ref{mu3}) gives $\partial_\mu u_0\cdot b=0$ and together with (\ref{mu2}), that is, $\partial_\mu u_0\times B=0$, they deliver $\partial_\mu u_0=0$. Dotting (\ref{mu3}) with $b$, we also find $\partial_\mu w_0=0$. Then (\ref{mu3}) reduces to $\partial_\mu u_1\times B=0$, which says that $\partial_\mu u_1=0$ up to trivial symmetries. Putting it all together, we conclude that $U$ is independent of $\mu$.
\end{proof}

\section{Relation to magnetohydrostatics}

So far quasisymmetry and symmetries in general of guiding-centre motion were treated independently of any other assumption on the magnetic field. In this section, we study approximate Hamiltonian symmetries in the presence of magnetohydrostatics (MHS),
\begin{equation}
\label{mhs}
J\times B=\nabla p,
\end{equation}
where $J=\curl B$ is the current density and $p$ is the scalar plasma pressure. This can be viewed as an extra restriction for the magnetic field that can be added to the previous symmetry conditions.

\begin{thm}
\label{appvsmhs}
For an MHS magnetic field with $dp\neq0$ a.e.~on $Q$ and density of irrational surfaces, every approximate Hamiltonian symmetry of FGC motion is an approximate quasisymmetry.
\end{thm}
\begin{proof}
First of all, write (\ref{mhs}) as $i_Bi_J\Omega=dp$ and note that $dB^\flat=i_J\Omega$. For any MHS field,
\begin{eqnarray}
\label{rel1}
L_BB^\flat&=i_BdB^\flat+di_BB^\flat=d(p+|B|^2),\\
i_{[J,B]}\Omega&=i_JL_B\Omega-L_Bi_J\Omega=-\,L_BdB^\flat=-\,dL_BB^\flat=0,
\end{eqnarray}
the latter implying $[J,B]=0$, since $\Omega$ is nondegenerate.

Now let $u=u_0+\epsilon u_1$ be the generator of an approximate Hamiltonian symmetry. By Lemma \ref{appvslem}, we have (\ref{vs11}) displayed here, $i_{u_0}i_B\Omega=d\psi_0$, from the first symmetry condition (\ref{vs1}), and by assumption $p=p(\psi_0)$. If $B$ is MHS, then
\begin{equation}
\label{rel2}
L_B(u_0\cdot B)=i_{u_0}L_BB^\flat=i_{u_0}dp+i_{u_0}d|B|^2=0,
\end{equation}
using $[u_0,B]=0$ from Theorem \ref{appvscons}, equation (\ref{rel1}), $i_{u_0}d\psi_0=0$ from (\ref{vs11}), and $i_{u_0}d|B|=0$ from the third symmetry condition (\ref{vs3}).

Next we are going to prove that $u_0$ is independent of $v_\pl$. To this end, note first that $J$ is tangent to flux surfaces and so
\begin{equation}
\label{rel3}
J=\kappa\,u_0 + \lambda B
\end{equation}
for some functions $\kappa,\lambda$. Crossing with $B$ gives $\kappa=-p'$. Applying $L_B$ to (\ref{rel3}), we deduce $L_B\lambda=0$, because $B$ commutes with $J$ and $u_0$, and $\kappa$ is a flux function. Taking the $v_\pl$-derivative of (\ref{rel3}), we get
\begin{equation}
\label{rel4}
\kappa\,\partial_{v_\pl}\!u_0 + (\partial_{v_\pl}\lambda)B=0,
\end{equation}
since $\kappa,J,B$ are spatial. Finally, the $v_\pl$-derivative of (\ref{rel2}) gives $\partial_{v_\pl}\lambda=0$ for $L_B|B|\neq0$. To see this, dot (\ref{rel4}) with $B$ and insert it, so
\begin{equation}
0=L_{\partial_{v_\pl}}L_B(u_0\cdot B)=L_BL_{\partial_{v_\pl}}(u_0\cdot B)=-\,\kappa^{-1}(\partial_{v_\pl}\lambda)L_B|B|^2,
\end{equation}
since $L_B\lambda=0$ and $L_B\kappa=0$. Substituting $\partial_{v_\pl}\lambda=0$ in (\ref{rel4}), we conclude $\partial_{v_\pl}\!u_0=0$ for $dp\neq0$.

By density of irrational surfaces, (\ref{rel2}) implies $L_{u_0}(u_0\cdot B)=0$ too. Likewise, $L_{u_0}\lambda=0$ from $L_B\lambda=0$. Thus, $C=u_0\cdot B$ and $\lambda$ are flux functions. Therefore we can write
\begin{equation}
\label{rel5}
L_{u_0}B^\flat=i_{u_0}dB^\flat+di_{u_0}B^\flat=i_{u_0}i_J\Omega+dC=(\lambda+C')d\psi
\end{equation}

Since $u_0$ is spatial, the conditions $i_{u_0}i_B\Omega=d\psi_0$, $[u_0,B]=0$ and $L_{u_0}|B|=0$ imply $u_0$ generates a local circle action on $Q$. See \cite{BKM}, Definition VIII.1 of a circle action and Theorem VIII.2(i) for a proof, as well as (59) there for the definition of circle-average. Circle-averaging (\ref{rel5}) gives then $0=(\lambda+C')d\psi$, because the average of $L_{u_0}$ of anything is zero and $\lambda,C$ are constant along $u_0$. Therefore $L_{u_0}B^\flat=0$.

Together with (\ref{vs3}), this gives $L_{u_0}b^\flat=0$ and so $X_0=0$, and then (\ref{vs2}) reduces to $L_{u_1}\beta=0$. Since $u_0$ is spatial, $\psi_1$ is too from Lemma \ref{appvslem}. Consequently, so is $u_1$ from (\ref{u1}), which completes the proof.
\end{proof}

\section{Discussion}

Compared to \cite{BKM}, Theorem \ref{appqs} shows that approximating quasisymmetry under the guiding-centre precision is to lowest order the same as exact quasisymmetry. While one might hope that the notion of quasisymmetry could be relaxed using approximate spatial (Hamiltonian) symmetries of guiding-centre motion, this theorem shows that it is impossible: if one insists that an approximate Hamiltonian symmetry is spatial then that symmetry must be a quasisymmetry. 
This is not much unexpected, since the quasisymmetry conditions were derived for all nonzero $q,m,\mu$ \cite{BKM}. Another way of seeing this is to note that $\epsilon$- and $v_\pl$-terms appear together in the Hamiltonian (or Lagrangian) formulation. Other spatial ways to approximate quasisymmetry could be perhaps more effective, as, for example, expansions near the magnetic axis.

Among the three conditions of quasisymmetry, $L_ub^\flat=0$ seems the most likely candidate to relax. Not included in earlier treatments, its necessity was first recognised in \cite{BQ}. Theorem \ref{appvs} with Lemma \ref{appvslem} say that an approximate phase-space Hamiltonian symmetry of guiding-centre motion does indeed weaken this condition to $v_\pl L_{u_0}b^\flat+i_{u_1}i_B\Omega=d\psi_1$. All the same, the remaining two conditions remain unchanged, providing flux surfaces and symmetric field strength. More explicitly, Theorem \ref{appvs2} shows that $L_{u_0}b^\flat$ is basically pushed back to the next-order term of the symmetry, the only restriction between $u_0$ and $\psi_1$ being their velocity dependence. Given Theorem \ref{appvscons} and Remark \ref{appvsconseq}, the arbitrariness of $L_{u_0}b^\flat$ is slightly limited to $i_bL_{u_0}b^\flat=0$, which is equivalent to $\divr u_0=0$ under the other symmetry conditions.

In conclusion, an approximate Hamiltonian $v_\pl$-symmetry generated by $u_0+\epsilon u_1$ satisfies the conditions (\ref{vs0}) to zero order and the first-order term is given by (\ref{u1})-(\ref{u1c2}). In the other direction, Theorem \ref{thm:weakvs} shows that a spatial vector field $u_0$ that satisfies (\ref{vs0}) for $L_{u_0}b^\flat\neq0$ with a second spatial flux function $\psi_1$ is the zeroth-order term of a weak quasisymmetry (Definition \ref{dfn:weakvs}) generator with $u_1$ given by (\ref{u1}). We may even extend this and say that a $v_\pl$-dependent vector field $u_0$ that satisfies (\ref{vs0}) and (\ref{u1c2}) for $L_{u_0}b^\flat\neq0$ and a $v_\pl$-dependent flux function $\psi_1$, is the zeroth-order term of an approximate Hamiltonian $v_\pl$-symmetry with $u_1$ given by (\ref{u1}).

Under the typical requirement of magnetohydrostatics though, every Hamiltonian $v_\pl$-symmetry is spatial and reduces to quasisymmetry again according to Theorem \ref{appvsmhs}.

The approximate constant of motion $K$ coming from an approximate Hamiltonian phase-space symmetry generalises the exact one derived for exact quasisymmetry in two ways. The first one is by introducing approximate flux surfaces via $\psi_1$ and the second one is its nonlinear character in $v_\pl$, when $u_0$ is not spatial. In any case, there is the question whether the first-order conserved quantity $K$ extends to higher orders leading to an adiabatic invariant. Repeating the symmetry analysis for approximate symmetries of higher order,  one imagines building $U$ and therefrom $K$ order by order. As the order increases, variations of $K$ would remain slow over larger time intervals. Ultimately, one would deduce an asymptotic series for $K$, which delivers variations of order $\epsilon$ over very long times, making it an adiabatic invariant assuming convergence.

Rodr\'iguez \textit{et al} \cite{RHB} introduced the notion of a weakly quasisymmetric magnetic field and argued that (for non-MHS fields) weak quasisymmetry implies that FGC motion admits (a) an approximate spatial symmetry, and (b) an approximate constant of motion. Their treatment does not consider first-order corrections to flux surfaces, i.e., they treat the subcase $\psi_1=0$. We partly agree with (b), but disagree with (a). While the weak quasisymmetry conditions for $\psi_1=0$ do indeed imply the existence of an approximate conserved quantity, namely $p$, (Proposition \ref{RHBi}), which is directly analogous to the case of exact quasisymmetry, for spatial $\psi_1\neq0$ they imply  a more general approximate conserved quantity, namely $K$ in (\ref{invariant}). Moreover, while the weak quasisymmetry conditions do imply the existence of an approximate symmetry for FGC motion, they do not imply the existence of an approximate spatial symmetry. Instead, our Proposition \ref{RHBs}, based on the one-to-one correspondence between symmetries and invariants, shows that the approximate symmetry associated with weak quasisymmetry even for zero $\psi_1$, namely $u_p$, acts non-trivially on both the guiding-centre position and parallel velocity, i.e., there is no way to regard the symmetry as operating in configuration space alone. Thus, Rodr\'iguez \textit{et al} correctly identify the conserved quantity associated with weak quasisymmetry for zero $\psi_1$, but incorrectly identify the infinitesimal generator of the corresponding phase-space symmetry.



It could be that there are no quasisymmetries (with bounded flux surfaces) other than axisymmetry. The quest for more general symmetries becomes then imperative as a means to relax the quasisymmetry notion. One such option is the longitudinal or second adiabatic invariant coming from a nonlocal symmetry, and the related concept of omnigeneity as a confinement condition (a sufficient one, but is it necessary?). A more direct generalisation, adopted here, was to allow symmetries on phase space instead of restricting to configuration space, and so involve the parallel velocity. Gyrosymmetry after all invokes the perpendicular velocity. Velocity-dependent symmetries could be of use, at least when it comes to guiding-centre integrability. Here we have made a first step of relaxing the requirement of a spatial symmetry by considering parallel-velocity-dependent symmetries within the approximate setup. Others could follow.

\ack

We would like to thank Eduardo Rodr\'iguez, Per Helander and Amitava Bhattacharjee for valuable discussions and helpful interaction. This work was supported by a grant from the Simons Foundation (601970, RSM) and by the Los Alamos National Laboratory LDRD program under project number 20180756PRD4.

\appendix

\section{Gyrosymmetry}\label{GSappendix}
Here we construct the gyrosymmetry generator $U$ up to zeroth-order terms from the magnetic moment $K$, using $i_U\omega=-\,dK$, i.e., we prove relations (\ref{gs1})-(\ref{gs2}).

Recall that the symplectic form of charged particle motion is $\omega=\beta+\epsilon\,dv\wedge dx$. 
As it will soon become apparent, $U$'s leading order is less by two than $K$'s, so let
\begin{eqnarray}
\label{K}
U&=U_0+\epsilon U_1+\cdots,\\
K&=\epsilon^2 K_2+\epsilon^3 K_3+\cdots.
\end{eqnarray}

First of all, note that for $i\geq0$ the $(i+1)$-th order terms of $i_U\omega=-\,dK$ split by $dv$ and $dx$ into
\begin{eqnarray}
\label{UK1}
U_i^x=\partial_vK_{i+1},\\
\label{UK2}
U_i^v+B\times U_{i+1}^x=-\,\partial_xK_{i+1},
\end{eqnarray}
where $U_i^x$ and $U_i^v$ are the spatial and velocity components of $U_i$, respectively. The second equation shows that the velocity components of $U_i$ are determined by the spatial ones of $U_{i+1}$ and the two combined together that in order to find $U_i$, we need both $K_{i+1}$ and $K_{i+2}$. In other words, a nonzero $K_{i+2}$ introduces a nonzero $U_i$, as we can see from
\begin{equation}
\label{UK3}
U_i^v=-\,\partial_xK_{i+1}-B\times\partial_vK_{i+2}
\end{equation}

Thus to build $U$ up to $U_1$, besides $K_2=v_\perp^2/2|B|$ we need $K_3$, as well. Following Littlejohn \cite{L1}, we take
\begin{equation}
\fl K_3=|B|^{-2}\left[K_2\,n\cdot\nabla|B|+v_\pl^2\,b\cdot\nabla b\cdot n+v_\pl\left(3\,n\cdot\nabla b\cdot v_\perp-v_\perp\cdot\nabla b\cdot n\right)\!/4\right],
\end{equation}
where $n=b\times v_\perp$. To apply (\ref{UK1})-(\ref{UK3}), we write down the required derivatives of $v_\pl$, $v_\perp$ and $n$ as functions of $x$ and $v$,
\begin{eqnarray}
\partial_{v_j}v_\pl=b^j,\\
\partial_{v_j}v_\perp=v_\perp^{-2}(v_\perp^jv_\perp+n^jn),\\
\partial_{v_j}n=v_\perp^{-2}(v_\perp^jn-n^jv_\perp),\\
\nabla(v_\perp^2)=-2v_\pl(v\times c+v\cdot\nabla b),
\end{eqnarray}
where $j=1,2,3$ and $c=\curl b$. For the forthcoming calculations, note also the vector identities
\begin{eqnarray}
&b\cdot\nabla|B|=-|B|\,\divr b,\\
&v_\perp\cdot\nabla b\cdot n-n\cdot\nabla b\cdot v_\perp=v_\perp^2\,b\cdot c,\\
&v_\perp\cdot\nabla b\cdot v_\perp+n\cdot\nabla b\cdot n=v_\perp^2\,\divr b.
\end{eqnarray}

Now, (\ref{UK1}) and (\ref{UK3}) for $i=0$ give (\ref{gs1}), and (\ref{UK1}) for $i=1$ gives the spatial component of (\ref{gs2}). For the velocity part, (\ref{UK3}) for $i=1$ yields
\begin{eqnarray}
\nonumber\fl U_1^v&=|B|^{-1}\!\left\{K_2\nabla|B|+v_\parallel(v\times c+v\cdot\nabla b)\right.\\
\nonumber\fl &\quad-v_\perp^{-2}\left[K_2\,v_\perp\cdot\nabla|B|+v_\pl^2\,b\cdot\nabla b\cdot v_\perp+v_\pl\left(v_\perp\cdot\nabla b\cdot v_\perp-n\cdot\nabla b\cdot n\right)\!/2\right]v_\perp\\
\fl &\quad\left.-\,v_\perp^{-2}\left[3K_2\,n\cdot\nabla|B|+v_\pl^2\,b\cdot\nabla b\cdot n+v_\pl\left(3\,n\cdot\nabla b\cdot v_\perp-v_\perp\cdot\nabla b\cdot n\right)\!/2\right]n\right\}\\
\nonumber\fl &=|B|^{-1}\!\left\{-2K_2(n\cdot\nabla|B|)n-v_\perp^2(\divr b)b/2+v_\pl v\times c+v_\pl(b\cdot c)n\right.\\
\fl&\qquad\quad\left.+\,v_\pl\left[v_\perp\cdot\nabla b+(n\cdot\nabla b)\times b\right]\!/2\right\}\\
\nonumber\fl&=|B|^{-1}\!\left\{-2K_2(n\cdot\nabla|B|)n+v_\pl v\times c+v_\pl(b\cdot c)n\right.\\
\fl&\qquad\quad\left.+\left[\left(b\times(v_\perp\cdot\nabla b)\right)\times v+(n\cdot\nabla b)\times v\right]\!/2\right\}.
\end{eqnarray}




\section{Vector calculus formulation}\label{VCappendix}
Below we give some of the expressions in Theorems \ref{appqs}-\ref{appvslem} in vector calculus notation. As before, $X_0=\curl b\times u_0 + \nabla(u_0\cdot b)$ and $\partial_{v_\pl}$ is short for $\partial/\partial_{v_\pl}$.


\begin{table}[ht]
\caption{Expressions in vector calculus.}
\footnotesize
\centering
\begin{tabular}{ll}
\br
Differential forms\quad&Vector calculus\\
\br
$L_{u_0}\beta=0$&$\curl(B\times u_0)=0$ and $\partial_{v_\pl}(B\times u_0)=0$\\
\mr
$L_{u_1}\beta+d(v_\pl L_{u_0}b^\flat)=0$ & $\curl(B\times u_1+v_\pl X_0)=0$ and $\partial_{v_\pl}(B\times u_1+v_\pl X_0)=\nabla[v_\pl\partial_{v_\pl}(u_0\cdot b)]$\\
\mr
$L_{u_0}|B|$&$u_0\cdot\nabla|B|$\\
\mr
$[u_0,B]$&$(u_0\cdot\nabla)B-(B\cdot\nabla)u_0$\\
\mr
$i_bL_{u_0}b^\flat$ & $b\cdot X_0$\\
\mr
$L_{u_0}b^\flat=i_Bi_{\partial_{v_\pl}\!u_1}\Omega$ & $X_0=\partial_{v_\pl}\!u_1\times B$\\
\mr
$i_{u_0}i_B\Omega=d\psi_0$&$B\times u_0=\nabla\psi_0$\\
\mr
$i_{u_1}i_B\Omega+v_\pl L_{u_0}b^\flat=d\psi_1$ & $B\times u_1+v_\pl X_0=\nabla\psi_1$ and $v_\pl\partial_{v_\pl}(u_0\cdot b)=\partial_{v_\pl}\psi_1$\\
\br
\end{tabular}
\end{table}


\noappendix


\section*{References}
\begin{thebibliography}{70}

\bibitem{B} Boozer A H 1983 Transport and isomorphic equilibria \textit{Phys. Fluids} \textbf{26} 496--99

\bibitem{NZ} N\"uhrenberg J and Zille R 1988 Quasihelically symmetric toroidal stellarators \PL A \textbf{129} 113--17

\bibitem{BKM} Burby J W, Kallinikos N and MacKay R S 2020 Some mathematics for quasi-symmetry \JMP \textbf{61} 093503

\bibitem{RHB} Rodr\'iguez E, Helander P and Bhattacharjee A 2020 Necessary and sufficient conditions for quasisymmetry \textit{Phys. Plasmas} \textbf{27} 062501

\bibitem{M} MacKay R S 2020 Tutorial on differential forms for plasma physics \textit{J. Plasma Phys.} \textbf{86} 925860101

\bibitem{L1} Littlejohn R G 1983 Variational principles of guiding centre motion \textit{J. Plasma Phys.} \textbf{29} 111--25

\bibitem{L2} Littlejohn R G 1984 Geometry and guiding center motion \textit{Fluids and Plasmas: Geometry and Dynamics} (\textit{Contemporary Mathematics} vol 28) ed. Marsden J E (Providence: AMS) p 151--67

\bibitem{BGI} Baikov V A, Gazizov R K and Ibragimov N H 1989 Approximate symmetries \textit{Math. USSR Sbornik} \textbf{64} 427--41

\bibitem{CG1} Cicogna G and Gaeta G 1994 Approximate symmetries in dynamical systems \textit{Il Nuovo Cimento} \textbf{109} 989--1008

\bibitem{IK} Ibragimov N H and Kovalev V F 2009 \textit{Approximate and Renormgroup Symmetries} (\textit{Nonlinear Physical Science}) (Beijing: Higher Education Press and Berlin, Heidelberg: Springer-Verlag)

\bibitem{CG2} Cicogna G and Gaeta G 1999 \textit{Symmetry and Perturbation Theory in Nonlinear Dynamics} (\textit{Lecture Notes in Physics} vol 57) (Berlin, Heidelberg: Springer-Verlag)

\bibitem{O} Olver P J 1993 \textit{Applications of Lie Groups to Differential Equations} (\textit{Graduate Texts in Mathematics} vol 107) 2nd edn (New York: Springer-Verlag)

\bibitem{FP} Ferrario C, Passerini A 1990 Symmetries and constants of motion for constrained Lagrangian systems: a presymplectic version of the Noether theorem \JPA \textbf{23} 5061--81

\bibitem{LD} de L\'eon M, de Diego D 1996 Symmetries and constants of the motion for singular Lagrangian systems \textit{Int. J. Theor. Phys.} \textbf{35} 975--1011

\bibitem{JS} Burby J W and Squire J 2020 General formulas for adiabatic invariants in nearly-periodic Hamiltonian systems arXiv:2005.00634

\bibitem{BQ} Burby J W and Qin H 2013 Toroidal precession as a geometric phase \textit{Phys. Plasmas} \textbf{20} 012511

\end{thebibliography}
\end{document}